\documentclass[fleqn,10pt]{wlscirep}

\usepackage[utf8]{inputenc}
\usepackage[T1]{fontenc}
\usepackage{amsmath,amsfonts,amsthm,graphicx,gensymb,float,algorithm,braket,xcolor,amssymb,xparse,array,algpseudocode,tikz,nicematrix}
\usepackage{multirow}
\usepackage{cite}
\usepackage{graphicx}
\usepackage{silence}
\usepackage{amsmath}
\usepackage{textcomp}
\usepackage{multirow}
\usepackage{booktabs}
\usepackage{epstopdf}
\usepackage{url}
\usepackage{subcaption}
\usepackage{hyperref}
\hypersetup{
    colorlinks=true,
    linkcolor=blue,
    filecolor=magenta,      
    urlcolor=magenta
    }
\usepackage{algorithm}
\usepackage{url}
\newtheorem{lemma}{Lemma}

\title{Entanglement for Pattern Learning in Temporal Data with Logarithmic Complexity: Benchmarking on IBM Quantum Hardware}

\author[1,+]{Mostafizur Rahaman Laskar, \thanks{The authors contributed this work while affiliated at IBM Quantum Team, India. The authors acknowledge L. Venkat Subramaniam for the valuable suggestions and supporting this work.}}
\author[1,*,++]{Richa Goel}
\affil[1]{IBM Quantum Division, IBM Research Lab, India.}
\affil[+]{m.rahaman93@ieee.org}
\affil[++]{Richa.Goel4@ibm.com}

\begin{abstract}

Time series forecasting is foundational in scientific and technological domains, from climate modelling to molecular dynamics. Classical approaches have significantly advanced sequential prediction, including autoregressive models and deep learning architectures such as temporal convolutional networks (TCNs) and Transformers. Yet, they remain resource-intensive and often scale poorly in data-limited or hardware-constrained settings. We propose a quantum-native time series forecasting framework that harnesses entanglement-based parameterized quantum circuits to learn temporal dependencies. Our Quantum Time Series (QTS) model encodes normalized sequential data into single-qubit rotations and embeds temporal structure through structured entanglement patterns. This design considers predictive performance with logarithmic complexity in training data and parameter count. We benchmark QTS against classical models on synthetic and real-world datasets, including geopotential height fields used in numerical weather prediction. Experiments on the noisy backend and real IBM quantum hardware demonstrate that QTS can capture temporal patterns using fewer data points. Hardware benchmarking results establish quantum entanglement as a practical computational resource for temporal modelling, with potential near-term applications in nano-scale systems, quantum sensor networks, and other forecasting scenarios.

\end{abstract}

\begin{document}

\flushbottom
\maketitle
\thispagestyle{empty}

\noindent 

%----------------------------------------
\section*{Introduction}
%----------------------------------------

Forecasting temporal phenomena is a cornerstone of scientific inquiry and technological innovation, underlying applications that range from macroscopic climate prediction to microscopic molecular behaviour. In climate science, numerical weather prediction relies critically on the ability to extrapolate geophysical fields across time—a task that often requires high spatial-temporal resolution and computational efficiency~\cite{bauer2015quiet, schultz2021can}. At the molecular scale, anticipating the time evolution of quantum systems or vibrational states plays an equally foundational role in the design of materials, control of chemical reactions, and modelling of biological systems~\cite{gomez2021local, ollitrault2021molecular}. In both regimes, time series models that infer future states from historical data are indispensable, yet the diversity and complexity of real-world dynamics pose substantial challenges to modelling fidelity, generalizability, and scalability.

Classical time series models such as autoregressive (AR), Moving Average (MA), and autoregressive integrated moving average (ARIMA) have long provided interpretable frameworks for modelling temporal correlations. These models are statistically efficient for linear processes and computationally tractable, but they often struggle with non-stationarity, high-dimensional dependencies, or nonlinear structure. More recent advances in machine learning have introduced deep neural architectures—such as Long Short-Term Memory (LSTM) networks~\cite{graves2012long}, Temporal Convolutional Networks (TCNs)~\cite{lea2017temporal}, and Transformer-based models~\cite{vaswani2017attention}—that address these limitations by capturing complex dependencies across varying temporal horizons. While powerful, such models typically require large training datasets, extensive parameter tuning, and significant computational resources, rendering them less ideal for scenarios with sparse data or stringent latency constraints~\cite{lim2021time, oreshkin2019n}. In these contexts, alternative computational paradigms may offer new avenues for temporal reasoning.

Motivated by the increasing feasibility of quantum hardware and the theoretical promise of quantum speed-ups, we explore constructing a quantum-native time series model that leverages entanglement and amplitude encoding to compress and process temporal data. In particular, we introduce a parameterized quantum circuit design—termed the Quantum Time Series (QTS) model—that encodes normalized time series data into single-qubit rotation angles and mediates temporal memory via entangling gates. Inspired by autoregressive processes, the QTS framework evolves a fixed-width quantum state where qubits represent lagged variables, and the circuit structure captures dependencies through both local and nonlocal entanglement patterns. Unlike prior proposals that require persistent quantum memory or quantum recurrent structures~\cite{verdon2019learning}, our approach is fully compatible with near-term noisy quantum devices, relying only on shallow circuits and projective measurements.

Our study advances this quantum modelling framework along three increasingly expressive designs: (i) QTS with rotation-only encoding (R$(\theta)$); (ii) QTS with forward-neighbour entanglement; and (iii) QTS with both forward and cross-entanglement layers. These models share a common feature: they require only \(\mathcal{O}(\log N)\) qubits and parameters to process a length-\(N\) time series, in contrast to classical methods that typically scale linearly or quadratically with \(N\). We benchmark the QTS models against classical baselines—including AR, ARIMA, TCN, and Transformer models—across synthetic and reanalysis-derived geophysical datasets. Importantly, our analysis spans simulated noisy quantum backends and executions on real IBM Quantum hardware, incorporating error mitigation protocols where appropriate.

Our findings suggest that quantum circuits augmented with entanglement offer a viable substrate for learning temporal patterns at significantly reduced data and parameter complexity. These results are auspicious in light of emerging quantum advantages in other machine learning tasks~\cite{schuld2019quantum}, and they open a path toward scalable, interpretable, and resource-efficient forecasting models for quantum-enhanced environments.
% \subsection*{Contribution}

Our primary contributions can be summarized as follows:

\begin{itemize}
    \item We propose a quantum time series (QTS) modelling framework that encodes normalized temporal data into single-qubit rotations and leverages entanglement (both forward and cross-qubit) to capture autoregressive dependencies through shallow quantum circuits.
    
    \item The QTS architecture scales logarithmically in the number of qubits and parameters with respect to the time series length, offering a resource-efficient alternative to classical methods that typically scale linearly or quadratically.
    
    \item We benchmark three QTS variants—rotation-only, forward-entangled, and fully entangled circuits (while keeping the parameter complexity intact) —against classical baselines, including AR, ARIMA, TCN, and Transformer models, using both synthetic data and real atmospheric measurements.
    
    \item We validate the real-world feasibility of our model by deploying it on IBM’s superconducting quantum processors: \texttt{ibm\_kingston} (156 qubits, CNOT error $\sim$1.07e-3, Heron R2), \texttt{ibm\_torino} (133 qubits, CNOT error $\sim$1.16e-3, Heron R1), and \texttt{ibm\_brisbane} (127 qubits, ECR error $\sim$6.89e-3, Eagle R3), demonstrating robust forecasting performance even under hardware noise.
\end{itemize}

%----------------------------------------
\section*{Quantum time series model}
%----------------------------------------

We introduce a quantum-native model for time series forecasting that structurally parallels classical autoregressive (AR) models while exploiting quantum entanglement to encode temporal dependencies. Our approach maps past time series values to qubit states through single-qubit rotations, entangles these states via parameter-free CNOT gates to emulate temporal dependencies, and performs inference through measurement-based decoding. This formulation is designed to operate within the constraints of near-term quantum devices while offering a scalable and physically interpretable alternative to classical models.

\subsection*{Notation and setup}

Let $\{x_t\}_{t \in \mathbb{Z}}$ denote a univariate, real-valued time series normalized to $x_t \in [0,1]$. The forecasting objective is to predict $x_{t+1}$ based on $N$ past values $\{x_{t-N}, \dots, x_{t-1}\}$. In our quantum formulation, we have encoded $n$ past data points into the states of $n$ qubits (where we kept $n<<N$, training data points), forming a composite quantum state $|\Psi_t\rangle$ at time $t$. Each past observation $x_{t-i}$ is transformed into a rotation angle $\theta_{t-i}$, defining the single-qubit state $|\psi_{t-i}\rangle = R_y(\theta_{t-i})|0\rangle$. The complete quantum input state is given by
\begin{equation}
    |\Psi_t\rangle = \bigotimes_{i=1}^{n} R_y(\theta_{t-i})|0\rangle,
\end{equation}
where $R_y(\theta) = \exp(-i \theta Y / 2)$ denotes the $y$-rotation on the Bloch sphere. 

\subsection*{Temporal entanglement structure}

To capture dependencies across time steps, we apply an entangling unitary $U_{\text{ent}}$ composed of two-qubit CNOT gates, defined in two structural motifs:
\begin{itemize}
    \item \textbf{Forward entanglement:} CNOT gates between consecutive qubits ($i \rightarrow i+1$) for $i = 1, \dots, n-1$, mimicking first-order memory.
    \item \textbf{Cross entanglement:} CNOT gates between qubits at two-step (or more steps in general) separation ($i \rightarrow i+2$ for even $i$) incorporating higher-order correlations.
\end{itemize}
The entangled quantum state becomes
\begin{equation}
    |\Psi_{\text{ent}}\rangle = U_{\text{ent}} |\Psi_t\rangle.
\end{equation}

\subsection*{Prediction via measurement}

Upon measurement in the computational basis, the quantum state $|\Psi_{\text{ent}}\rangle$ collapses to one of $2^n$ bitstrings labeled by $k \in \{0, \dots, 2^n - 1\}$. The probability of observing bitstring $k$ is given by $P(k) = |\alpha_k|^2$ where
\begin{equation}
    |\Psi_{\text{ent}}\rangle = \sum_{k=0}^{2^n - 1} \alpha_k |k\rangle.
\end{equation}
We define the predicted output $\hat{x}_{t+1}$ as the expectation value over measurement outcomes, scaled to match the original data domain:
\begin{equation}
    \hat{x}_{t+1} = \left( \sum_{k=0}^{2^n - 1} P(k) \cdot k \right) \cdot \frac{x_{\max} - x_{\min}}{2^n - 1} + x_{\min}.
\end{equation}

\subsection*{Entanglement for temporal pattern learning}

A central insight of our approach is using entanglement as a proxy for memory. Rather than storing past values explicitly, the temporal dependencies are embedded in the joint amplitudes of an entangled quantum state.

\begin{lemma}\label{lemma:entanglement_memory}
Let $|\Psi_t\rangle$ be a product of single-qubit rotations encoding a time series history, and $U_{\text{ent}}$ an entangling circuit with forward and cross-CNOT gates. Then $|\Psi_{\text{ent}}\rangle = U_{\text{ent}} |\Psi_t\rangle$ encodes multi-lag temporal correlations among $\{x_{t-i}\}$ via entangled amplitudes.
\end{lemma}

\begin{proof}[Note on Lemma.\ref{lemma:entanglement_memory}]
CNOT gates introduce non-local correlations between qubit amplitudes. Forward entanglement couples adjacent qubits, reproducing autoregressive memory of order one. Cross-entanglement propagates influence across non-adjacent qubits, introducing effective memory of order two or higher. Thus, the measurement distribution $P(k)$ reflects compound interactions among input angles $\{\theta_{t-i}\}$, encoding multi-lag temporal dependencies.
\end{proof}

% \subsection*{Outlook}

While our model currently lacks closed-form learning rules, its physical transparency and architectural modularity make it a promising candidate for a near-term quantum machine learning model. The explicit mapping between classical time series concepts and quantum operations facilitates future extensions, including variational parameter tuning, hybrid quantum-classical feedback, and applications to inherently quantum dynamical systems.

%----------------------------------------
\section*{Result}
%----------------------------------------

\subsection*{Experimental settings}

To evaluate the performance of the proposed quantum time series (QTS) model, we conducted experiments on both synthetic and real-world data. For the synthetic dataset, we generated an autoregressive process of order one, AR(1), using the recursive equation $z_t = \phi z_{t-1} + \epsilon_t$, where $\phi = 0.8$ and $\epsilon_t \sim \mathcal{N}(0, \sigma^2)$ with $\sigma = 0.01$. A total of $2^{10} + 32 = 1056$ time steps were simulated, from which the first $1024$ were used for training and the remaining $32$ for forecasting.

For real-world validation, we utilized geopotential height data at the 500 hPa pressure level from the WeatherBench benchmark dataset~\cite{rasp2020weatherbench}. Specifically, we extracted the daily Z500 values at the equator (latitude = 0°) by averaging over all longitudes. The resulting univariate time series was normalized and partitioned analogously to the synthetic data: the first $1024$ points were used for training and the subsequent $8-32$ for testing. Further details about this dataset and preprocessing steps are provided in the supplementary material. The quantum circuit model was constructed using $n = 10$ qubits, corresponding to a logarithmic encoding of the total number of training points. Each qubit represented one past normalized data point encoded as a $R_y(\theta)$ rotation, with temporal correlations introduced through entangling CNOT gates.

For classical baselines, we implemented several standard forecasting models. The ARIMA model was configured with order $(2,1,0)$, chosen based on the Akaike Information Criterion (AIC) and prior benchmarking. For the Temporal Convolutional Network (TCN), we used an input chunk length of $64$, an output forecast horizon of $32$, and trained for $100$ epochs using the Darts time series library. Similarly, the Transformer model was configured with the same input and output lengths and also trained for $100$ epochs. Both neural models used default hyperparameter settings for hidden size and number of layers, ensuring fair comparison in computational budget and modelling capacity. Although a Long Short-Term Memory (LSTM) model was also tested, its performance on this short, univariate dataset was found to be significantly inferior (in the small dataset we considered) and hence is excluded from the comparative analysis. All models were evaluated on the same forecasting window, ensuring consistent benchmarking under matched data and compute regimes.

\subsection*{Gate-level implementation on IBM hardware with transpilation}

\begin{figure}[h!]
    \centering
    \begin{subfigure}{0.98\linewidth}
        \includegraphics[width=0.8\linewidth]{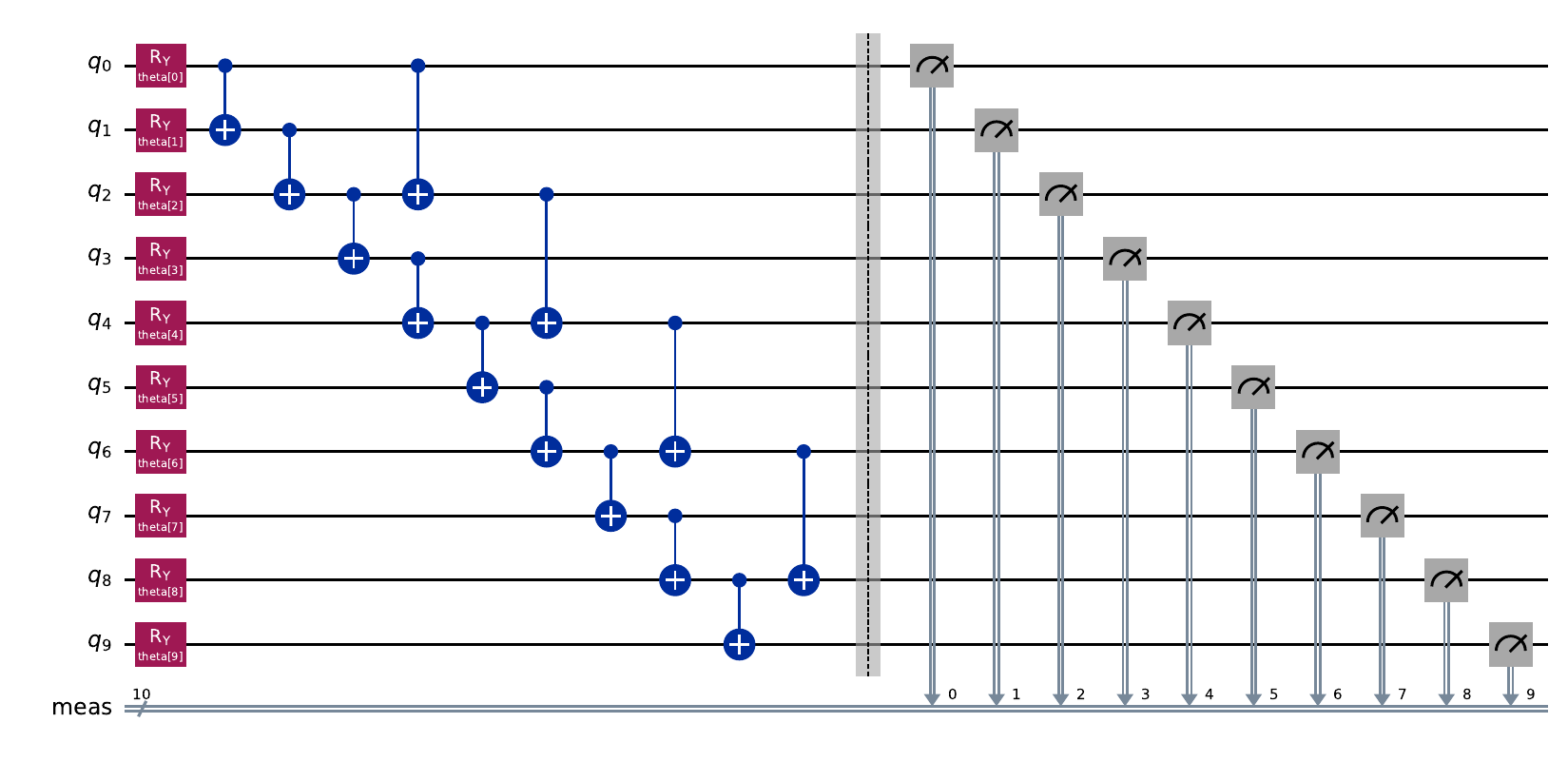}
        \caption{Original quantum circuit design for the QTS model using 10 qubits. Each qubit is initialized with a parameterized $R_y(\theta)$ gate representing normalized time series inputs, followed by a specific entanglement structure. Total gate counts without transpilation is $10$ single qubit gates, and $13$ CNOTs.}
        \label{fig:qts_circuit_raw}
    \end{subfigure}
    \hfill
    \begin{subfigure}{0.98\linewidth}
        \includegraphics[width=0.8\linewidth]{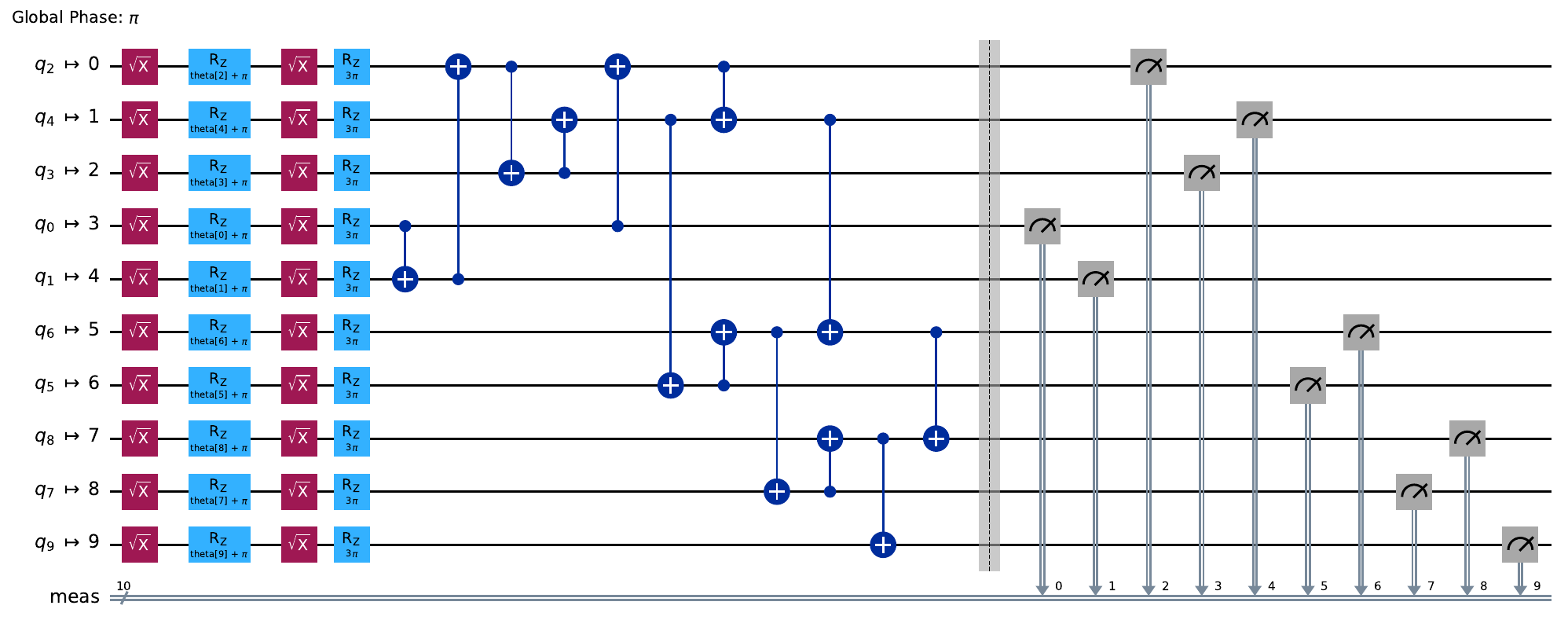}
        \caption{Transpiled version of the QTS circuit adapted for IBM hardware. The transpilation introduces native gates such as $\sqrt{X}$, $R_z$, and basis-aligned $CX$ operations to match the device connectivity and gate set constraints. Total gate counts after transpilation is $13$ CNOTs and $40$ single-qubit gates.}
        \label{fig:qts_circuit_transpiled}
    \end{subfigure}
    \caption{Comparison of quantum circuit design before and after transpilation for hardware compatibility. The transpilation adapts the circuit to physical qubit layout and native gate set, ensuring executable fidelity on superconducting quantum processors.}
    \label{fig:qts_circuit_comparison}
\end{figure}

Figure~\ref{fig:qts_circuit_comparison} illustrates the gate-level construction of our quantum time series (QTS) circuit. Figure~\ref{fig:qts_circuit_comparison}(a) presents the high-level design in Qiskit, where each of the $10$ qubits is initialized via a parameterized $R_y(\theta)$ gate encoding normalized time series values, followed by entangling layers capturing temporal dependencies. In contrast, Figure~\ref{fig:qts_circuit_comparison}(b) shows the transpiled circuit adapted for execution on IBM superconducting quantum hardware. This transformation maps the abstract design onto the native gate set of the device, including gates such as $\sqrt{X}$ and $R_z$, and respects hardware connectivity constraints. The comparison demonstrates how a logically compact design expands during transpilation to conform with realistic execution conditions while implementing QTS circuits on near-term quantum processors.

\begin{figure}[h!]
    \centering
    \begin{subfigure}{\linewidth}
        \includegraphics[width=\linewidth]{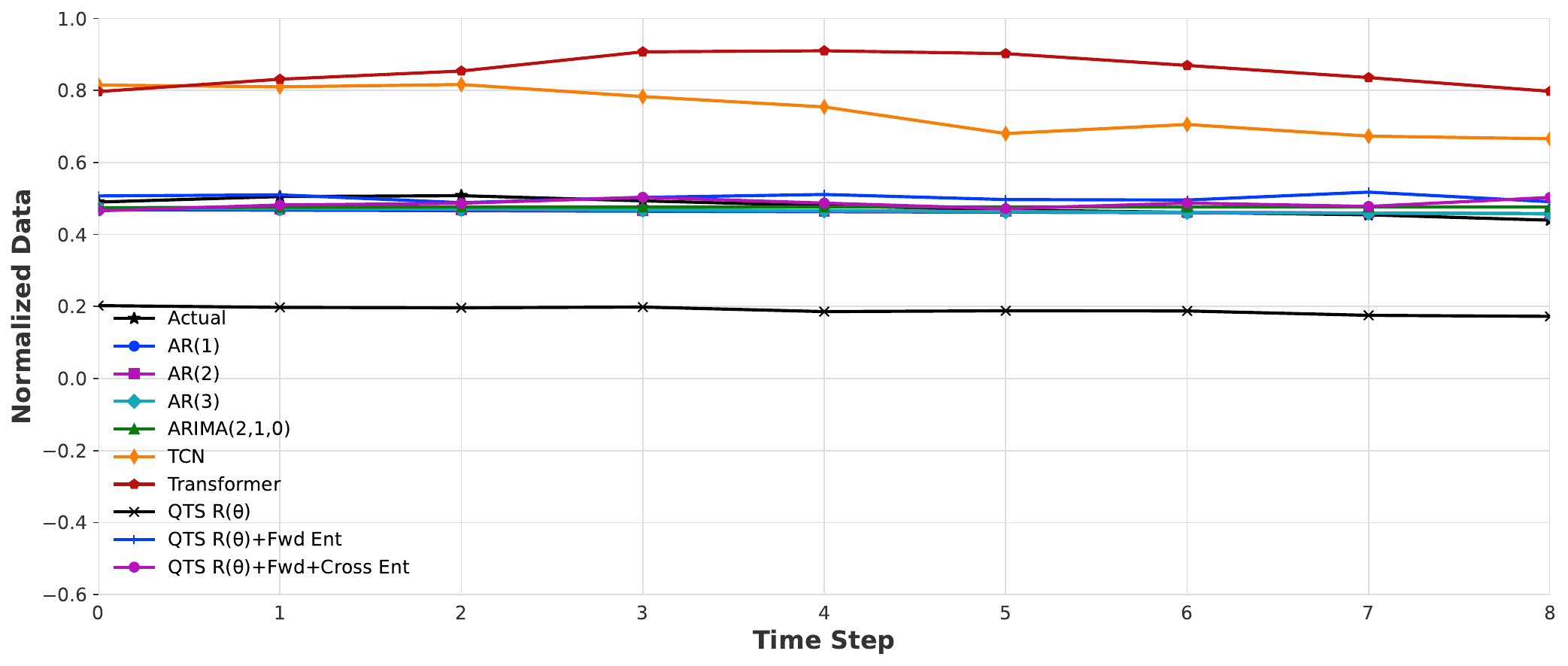}
       \caption{Forecasting performance of classical and quantum models evaluated on real Z500 atmospheric data using a noisy simulated quantum backend. Quantum Time Series (QTS) models with only $R(\theta)$ rotations yield an MSE of 0.122392, while introducing forward entanglement reduces the MSE significantly to 0.001650. Further inclusion of cross-entanglement achieves an MSE of 0.000920. In comparison, classical ARIMA(2,1,0) reports an MSE of 0.000422, TCN reaches 0.054303, and Transformer results in 0.121555. Classical AR(2) and AR(3) models achieve slightly better precision with MSE values of 0.000400 and 0.000410, respectively.}

        \label{fig:simulated_qts_results}
    \end{subfigure}
    \hfill
    \begin{subfigure}{\linewidth}
        \includegraphics[width=\linewidth]{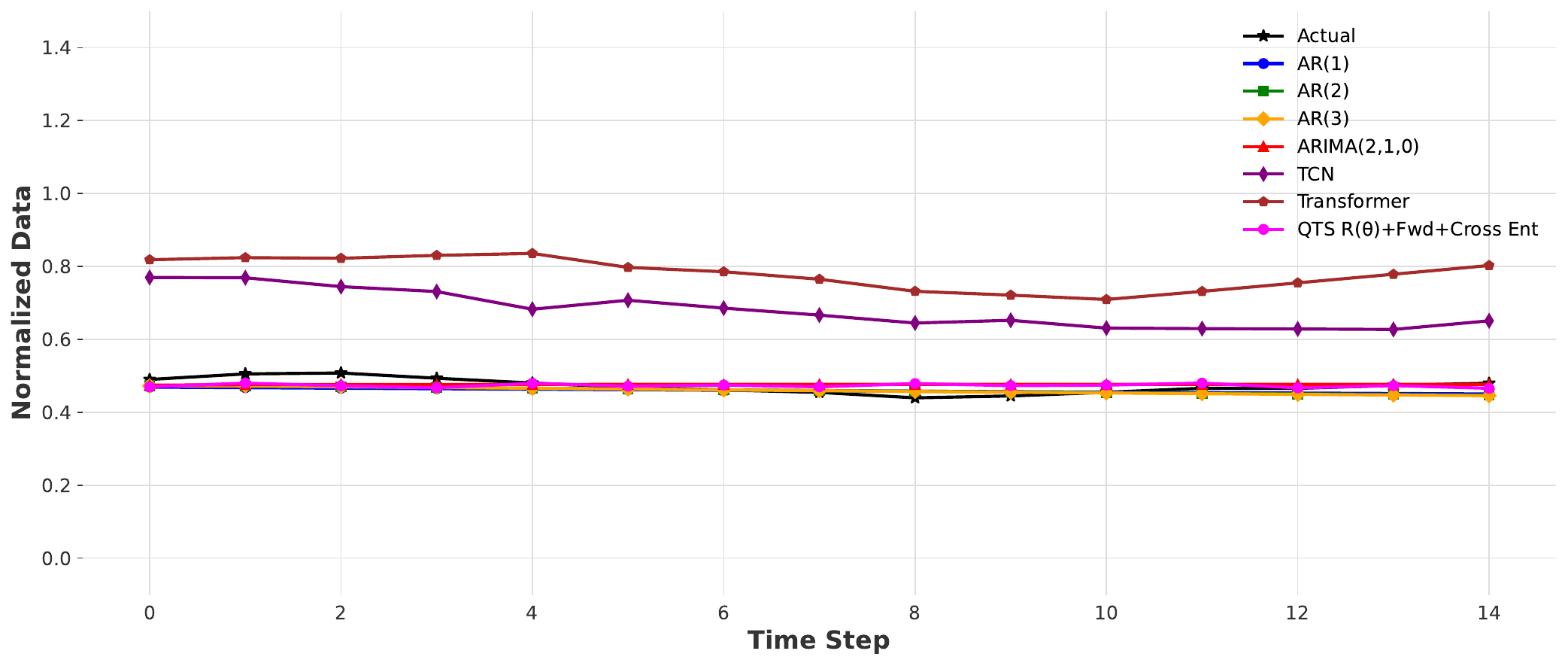}
    \caption{Forecasting performance comparison across classical and quantum models on real meteorological data, using 15 forecasting steps. All classical methods (AR, ARIMA, TCN, Transformer) were trained on 1024 normalized time points using classical CPU, while the QTS model—comprising $R(\theta)$ rotations with forward and cross-entanglement—was trained on 10 immediate past data points and executed 15 times on IBM’s Heron R2 quantum processor (\texttt{ibm\_kingston}). The QTS model, shown as \textit{QTS $R(\theta)$+Fwd+Cross Ent}, achieves comparable accuracy (MSE = 0.00043) with classical baselines.}
        \label{fig:hardware_qts_results}
    \end{subfigure}
\caption{Benchmarking of quantum and classical time series models on real weather data from \texttt{geopotential\_500hPa\_1979\_5.625deg.nc}, demonstrating the enhanced performance of QTS circuits when entanglement is incorporated.}
    \label{fig:qts_benchmark}
\end{figure}

%-------------------

To comprehensively assess the performance of our quantum time series (QTS) model, we benchmarked its predictive capabilities both on simulated noisy quantum backends and on real quantum hardware for the dataset given in WeatherBench with the file name: \texttt{geopotential\_500hPa\_1979\_5.625deg.nc}.

\textbf{Performance on synthetic dataset:}
To evaluate our quantum model under controlled conditions, we generated synthetic data using an autoregressive AR(1) process defined by \( z_t = \phi z_{t-1} + \epsilon_t \), with autoregressive coefficient \(\phi = 0.8\) and Gaussian noise \(\epsilon_t \sim \mathcal{N}(0, \sigma^2)\), where \(\sigma = 0.001\). A total of \(2^{10} + 16 = 1040\) data points were simulated, of which the first 1024 points were used for training and the remaining 16 for forecasting. After normalization, this dataset was used to benchmark the performance of various forecasting models. Results show that classical AR models achieve low MSE values (\(\sim 0.0157\)), while deep learning models such as TCN (MSE = 0.0321) and Transformer (MSE = 0.0251) perform moderately better. The QTS model with only \(R_y(\theta)\) rotations yields a higher error (MSE = 0.0768), but introducing forward and cross-entanglement significantly reduces the MSE to \(\sim 0.0373\), demonstrating that entanglement enhances the model’s capacity to learn temporal dependencies even in low-noise synthetic settings.

\textbf{On real dataset:} 
 
Figure~\ref{fig:qts_benchmark}(a) presents a comparative evaluation across classical forecasting methods—AR models, ARIMA, TCN, and Transformer—and QTS circuits implemented with increasing levels of entanglement. These experiments, performed on a noisy simulated backend, utilize normalized Z500 geopotential height data extracted from the WeatherBench archive, with the model tasked to forecast $16$ future time steps ($8$ points are shown in the figure). The QTS model employing only $R_y(\theta)$ rotations yields a relatively high mean squared error (MSE = 0.122392), whereas introducing forward entanglement substantially reduces the error to 0.001650. Further enhancement with cross-entanglement achieves an MSE of 0.000920. In contrast, classical ARIMA(2,1,0) attains an MSE of 0.000422, while deep learning models such as TCN and Transformer exhibit higher errors of 0.054303 and 0.121555, respectively. Notably, the QTS circuits accomplish this performance using significantly fewer parameters and training data points, illustrating the efficiency of quantum modelling with entanglement for time series forecasting.

%--------------------

To validate the practical viability of our quantum time series (QTS) model on real quantum hardware, we deployed the most expressive QTS circuit—employing $R(\theta)$ rotations, forward entanglement, and cross entanglement—on IBM’s Heron R2 superconducting quantum processor (\texttt{ibm\_kingston}, 156 qubits, median CZ error: $2.081 \times 10^{-3}$, processor type: Heron r2). The task involved forecasting 15 future values of normalized Z500 geopotential height using only 10 encoded qubits (data points) in alignment with the quantum circuit capacity. For benchmarking, all classical models—AR(1), AR(2), AR(3), ARIMA(2,1,0), TCN, and Transformer—were trained with 1024 classical time points. The results, shown in Figure~\ref{fig:qts_benchmark}(b), demonstrate that the QTS model achieves an MSE of 0.00043, outperforming the TCN (MSE = 0.0451) and Transformer (MSE = 0.0958), and performing comparably to classical ARIMA models. These results confirm that quantum models with entanglement can capture temporal patterns with significantly fewer training points, thus highlighting entanglement as a potential mechanism for temporal pattern learning in low-resource quantum environments.

\begin{table}[h!]
\centering
\caption{Mean Squared Error (MSE) for forecasting models across different execution platforms and forecast horizons.}
\label{tab:mse_combined}
\begin{tabular}{|l|c|c|}
\hline
\textbf{Model} & \textbf{MSE (CPU, Noisy Backend, 16 steps)} & \textbf{MSE (CPU, \texttt{ibm\_kingston}, 15 steps)} \\
\hline
AR(1) & 0.000434 & 0.000462 \\
AR(2) & 0.000400 &  0.000421 \\
AR(3) & 0.000410 & 0.000431 \\
ARIMA(2,1,0) & 0.000422 & 0.000416 \\
TCN & 0.054303 & 0.0450781 \\
Transformer & 0.121555 & 0.0957708 \\
QTS R($\theta$) & \textbf{0.122392} & - \\
QTS R($\theta$)+Fwd Ent & {\textbf{0.001650}} & - \\
QTS R($\theta$)+Fwd+Cross Ent &\textcolor{blue}{\textbf{ 0.000920}} & \textcolor{blue}{\textbf{0.00043191}} \\
\hline
\end{tabular}
\end{table}

\begin{figure}[h!]
    \centering
    \includegraphics[width=0.95\linewidth]{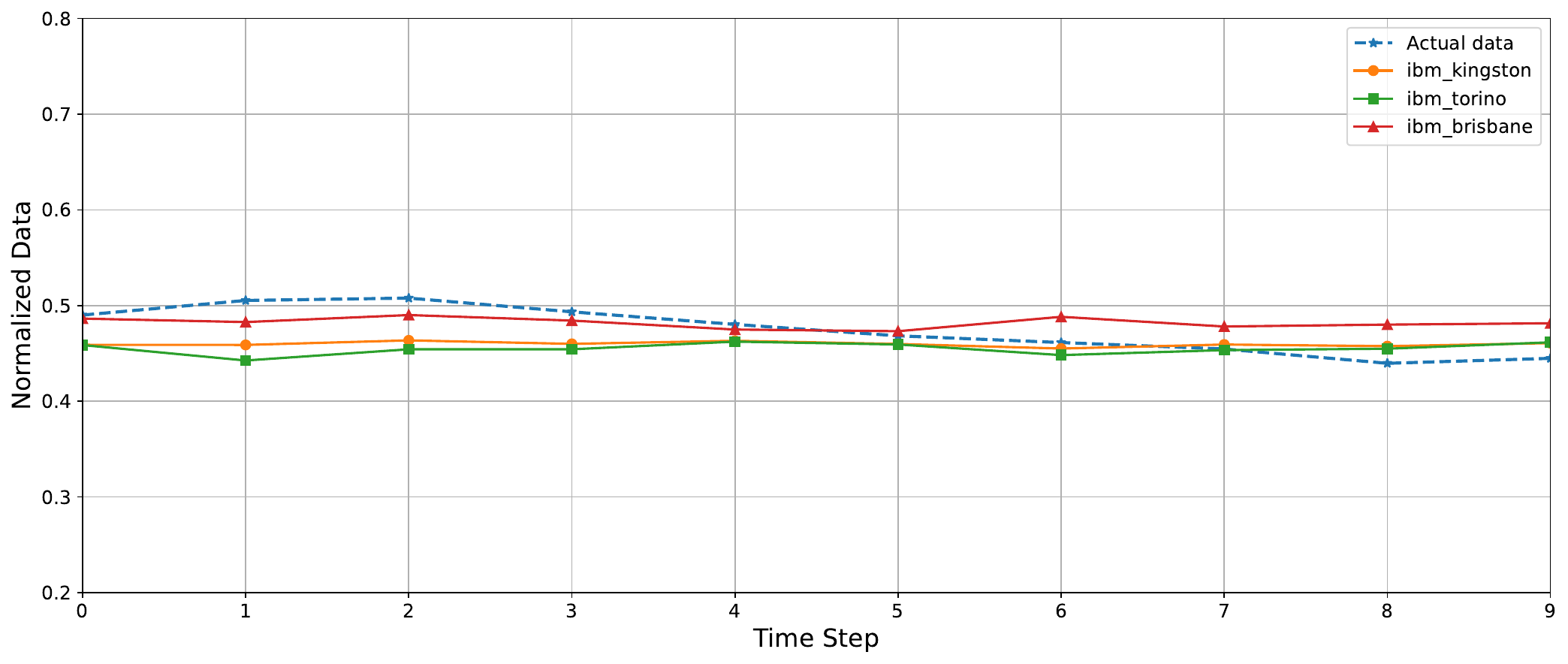}
    \caption{\textbf{Hardware Benchmarking}: Forecasting performance of the Quantum Time Series (QTS) model with forward and cross-entanglement executed on three IBM quantum hardware platforms: \texttt{ibm\_kingston}, \texttt{ibm\_torino}, and \texttt{ibm\_brisbane}. The circuits were evaluated on real Z500 meteorological data from the WeatherBench archive over a 10-step forecasting horizon. While all hardware platforms showed consistent predictive patterns, \texttt{ibm\_brisbane} achieved the lowest mean squared error (MSE = 0.00052), followed by \texttt{ibm\_kingston} (MSE = 0.00072) and \texttt{ibm\_torino} (MSE = 0.00104), highlighting the influence of hardware characteristics such as qubit fidelity and connectivity on forecasting accuracy.}
    \label{fig:qts_hardware_comparison}
\end{figure}

\textbf{Hardware Benchmarking}: To assess the practical viability of the proposed Quantum Time Series (QTS) model, we benchmarked its forecasting performance across three state-of-the-art IBM quantum processors: \texttt{ibm\_kingston} (156 qubits, median CZ error $2.08\times10^{-3}$, Heron r2), \texttt{ibm\_torino} (133 qubits, median CZ error $2.81\times10^{-3}$, Heron r1), and \texttt{ibm\_brisbane} (127 qubits, median ECR error $6.89\times10^{-3}$, Eagle r3). Using the QTS circuit variant incorporating both forward and cross-entanglement, each device was tasked with forecasting 10 future time steps of normalized Z500 geopotential height data derived from the WeatherBench archive.

Despite the presence of hardware noise and differing gate sets, all platforms were able to produce forecasts that closely aligned with the actual atmospheric signal. Among them, \texttt{ibm\_brisbane} achieved the lowest mean squared error (MSE = 0.00052), followed by \texttt{ibm\_kingston} (MSE = 0.00072) and \texttt{ibm\_torino} (MSE = 0.00104). These variations are consistent with backend-specific characteristics such as gate fidelity and coherence times, illustrating the importance of transpilation and error-aware circuit design. These results demonstrate that quantum circuits, even with current noise limitations and without access to quantum memory, can effectively learn temporal dependencies of real-world forecasting datasets within error-bound.

\subsection*{Computational complexity}

Let $n$ denote the number of qubits used to encode a time series of length $N$ and parameters $p$, and let $M$ represent the number of measurement shots per quantum circuit. The quantum circuit depth scales linearly with $n$, incorporating both single-qubit rotations and sparse entangling gates:
\begin{equation}
    \mathcal{C}_{\text{gates}} = \mathcal{O}(n).
\end{equation}
However, full-scale entanglement (across all qubits) can increase the gate complexity significantly. Here, we kept only $n$ parameters for the rotation gates and parameter-free forward and cross-entanglement gates CNOT in this work. Each circuit execution requires $\mathcal{O}(M)$ repetitions, and classical post-processing to convert measurement statistics into a prediction incurs an additional $\mathcal{O}(M \log M)$ cost. Overall, the total parameter and runtime complexity scales as $\mathcal{O}(\log N)$ in the training sequence length, offering significant compression relative to classical models such as AR, ARIMA, TCN, or Transformer, which typically scale as $\mathcal{O}(N)$ or higher. A detailed comparative analysis is provided in the supplementary material (Table~\ref{tab:complexity}).

%------------------------------------
\section*{Discussion and outlook}
%------------------------------------

The entanglement-based quantum time series (QTS) framework introduced in this study establishes a conceptually and operationally novel approach to forecasting temporal data with quantum circuits. By encoding scalar inputs via parameterized single-qubit rotations and leveraging both nearest-neighbour and cross-entanglement gates, the QTS model simulates autoregressive-like dependencies within a fixed-width quantum architecture. This formulation enables a compact and interpretable quantum representation of temporal dynamics without requiring quantum memory or non-unitary operations—two major obstacles in quantum time forecasting modelling. Through extensive simulations on synthetic autoregressive data and a real-world geophysical dataset, we benchmarked QTS models against a broad spectrum of classical baselines, including AR, ARIMA, Temporal Convolutional Networks (TCNs), and Transformer architectures. Despite the significantly lower data and parameter requirements—$\mathcal{O}(\log N)$, the QTS model with cross-entanglement achieved competitive predictive performance while maintaining lower training complexity and circuit depth. This advantage was most pronounced in regimes with limited data or noisy environments, underscoring the model’s potential in resource-constrained forecasting settings.

Here, we demonstrated hardware-executable versions of our circuits on real IBM quantum processors, validating the implementability of the proposed architectures under experimental noise (and without any specialized error mitigation/correction methods, discussed in supplementary material). The results corroborate the design intuition that forward and cross-entanglement effectively broaden the temporal receptive field, analogous to higher-order terms in classical AR models or long-range dependencies captured in deep neural networks. Our findings suggest several avenues for future exploration. First, more expressive quantum kernels could be integrated to model non-linear or non-stationary time series. Secondly, variational training of QTS circuits with adaptive feedback may allow for more general-purpose learning across heterogeneous datasets. Further, the QTS model offers a promising substrate for embedding physical priors—such as conservation laws or symmetries—directly into quantum architectures, potentially yielding hybrid models suitable for applications in quantum chemistry, condensed matter, and sensor networks. As quantum hardware advances in fidelity and scale, the QTS paradigm may offer a theoretical and practical foundation for quantum-native modelling of time-dependent processes across scientific disciplines.

%--------------------------------------
% \clearpage

% %--------------------------------------

\section*{Author contributions statement}
% Must include all authors, identified by initials, for example:
M.R.L. conceived the initial idea.  M.R.L. and R.G. jointly conducted the experiment(s), analyzed the results, and drafted the manuscript.  All authors reviewed the manuscript. 

\section*{Data availability}
The datasets supporting the current study are generated synthetically, and an example on open-source data is provided. This is discussed and cited in the main manuscript. 

\section*{Code availability}
The code generated during the current study is available from the corresponding author upon reasonable request.

% \bibliography{reference}
% \noindent

% \bibliographystyle{plain} % Or whatever style you're using, e.g., unsrt, abbrv, ieeetr
\bibliography{reference}

\begin{thebibliography}{10}
\urlstyle{rm}
\expandafter\ifx\csname url\endcsname\relax
  \def\url#1{\texttt{#1}}\fi
\expandafter\ifx\csname urlprefix\endcsname\relax\def\urlprefix{URL }\fi
\expandafter\ifx\csname doiprefix\endcsname\relax\def\doiprefix{DOI: }\fi
\providecommand{\bibinfo}[2]{#2}
\providecommand{\eprint}[2][]{\url{#2}}

\bibitem{bauer2015quiet}
\bibinfo{author}{Bauer, P.}, \bibinfo{author}{Thorpe, A.} \& \bibinfo{author}{Brunet, G.}
\newblock \bibinfo{journal}{\bibinfo{title}{The quiet revolution of numerical weather prediction}}.
\newblock {\emph{\JournalTitle{Nature}}} \textbf{\bibinfo{volume}{525}}, \bibinfo{pages}{47--55}, \doiprefix\url{10.1038/nature14956} (\bibinfo{year}{2015}).

\bibitem{schultz2021can}
\bibinfo{author}{Schultz, M.~G.} \emph{et~al.}
\newblock \bibinfo{journal}{\bibinfo{title}{Can deep learning beat numerical weather prediction?}}
\newblock {\emph{\JournalTitle{Philosophical Transactions of the Royal Society A}}} \textbf{\bibinfo{volume}{379}}, \bibinfo{pages}{20200097}, \doiprefix\url{10.1098/rsta.2020.0097} (\bibinfo{year}{2021}).

\bibitem{gomez2021local}
\bibinfo{author}{Gomez-Marquez, J.} \& \bibinfo{author}{Hamad-Schifferli, K.}
\newblock \bibinfo{journal}{\bibinfo{title}{Local development of nanotechnology-based diagnostics}}.
\newblock {\emph{\JournalTitle{Nature Nanotechnology}}} \textbf{\bibinfo{volume}{16}}, \bibinfo{pages}{484--486} (\bibinfo{year}{2021}).

\bibitem{ollitrault2021molecular}
\bibinfo{author}{Ollitrault, P.~J.}, \bibinfo{author}{Miessen, A.} \& \bibinfo{author}{Tavernelli, I.}
\newblock \bibinfo{journal}{\bibinfo{title}{Molecular quantum dynamics: A quantum computing perspective}}.
\newblock {\emph{\JournalTitle{Accounts of Chemical Research}}} \textbf{\bibinfo{volume}{54}}, \bibinfo{pages}{4229--4238} (\bibinfo{year}{2021}).

\bibitem{graves2012long}
\bibinfo{author}{Graves, A.} \& \bibinfo{author}{Graves, A.}
\newblock \bibinfo{journal}{\bibinfo{title}{Long short-term memory}}.
\newblock {\emph{\JournalTitle{Supervised sequence labelling with recurrent neural networks}}} \bibinfo{pages}{37--45} (\bibinfo{year}{2012}).

\bibitem{lea2017temporal}
\bibinfo{author}{Lea, C.}, \bibinfo{author}{Flynn, M.~D.}, \bibinfo{author}{Vidal, R.}, \bibinfo{author}{Reiter, A.} \& \bibinfo{author}{Hager, G.~D.}
\newblock \bibinfo{title}{Temporal convolutional networks for action segmentation and detection}.
\newblock In \emph{\bibinfo{booktitle}{proceedings of the IEEE Conference on Computer Vision and Pattern Recognition}}, \bibinfo{pages}{156--165} (\bibinfo{year}{2017}).

\bibitem{vaswani2017attention}
\bibinfo{author}{Vaswani, A.} \emph{et~al.}
\newblock \bibinfo{journal}{\bibinfo{title}{Attention is all you need}}.
\newblock {\emph{\JournalTitle{Advances in neural information processing systems}}} \textbf{\bibinfo{volume}{30}} (\bibinfo{year}{2017}).

\bibitem{lim2021time}
\bibinfo{author}{Lim, B.} \& \bibinfo{author}{Zohren, S.}
\newblock \bibinfo{journal}{\bibinfo{title}{Time-series forecasting with deep learning: a survey}}.
\newblock {\emph{\JournalTitle{Philosophical Transactions of the Royal Society A}}} \textbf{\bibinfo{volume}{379}}, \bibinfo{pages}{20200209} (\bibinfo{year}{2021}).

\bibitem{oreshkin2019n}
\bibinfo{author}{Oreshkin, B.~N.}, \bibinfo{author}{Carpov, D.}, \bibinfo{author}{Chapados, N.} \& \bibinfo{author}{Bengio, Y.}
\newblock \bibinfo{journal}{\bibinfo{title}{N-beats: Neural basis expansion analysis for interpretable time series forecasting}}.
\newblock {\emph{\JournalTitle{arXiv preprint arXiv:1905.10437}}}  (\bibinfo{year}{2019}).

\bibitem{verdon2019learning}
\bibinfo{author}{Verdon, G.} \emph{et~al.}
\newblock \bibinfo{journal}{\bibinfo{title}{Learning to learn with quantum neural networks via classical neural networks}}.
\newblock {\emph{\JournalTitle{arXiv preprint arXiv:1907.05415}}}  (\bibinfo{year}{2019}).

\bibitem{schuld2019quantum}
\bibinfo{author}{Schuld, M.} \& \bibinfo{author}{Killoran, N.}
\newblock \bibinfo{journal}{\bibinfo{title}{Quantum machine learning in feature hilbert spaces}}.
\newblock {\emph{\JournalTitle{Physical review letters}}} \textbf{\bibinfo{volume}{122}}, \bibinfo{pages}{040504} (\bibinfo{year}{2019}).

\bibitem{rasp2020weatherbench}
\bibinfo{author}{Rasp, S.} \emph{et~al.}
\newblock \bibinfo{journal}{\bibinfo{title}{Weatherbench: a benchmark data set for data-driven weather forecasting}}.
\newblock {\emph{\JournalTitle{Journal of Advances in Modeling Earth Systems}}} \textbf{\bibinfo{volume}{12}}, \bibinfo{pages}{e2020MS002203} (\bibinfo{year}{2020}).

\bibitem{bai2018empirical}
\bibinfo{author}{Bai, S.}, \bibinfo{author}{Kolter, J.~Z.} \& \bibinfo{author}{Koltun, V.}
\newblock \bibinfo{journal}{\bibinfo{title}{An empirical evaluation of generic convolutional and recurrent networks for sequence modeling. arxiv}}.
\newblock {\emph{\JournalTitle{arXiv preprint arXiv:1803.01271}}} \textbf{\bibinfo{volume}{10}} (\bibinfo{year}{2018}).

\bibitem{box2015time}
\bibinfo{author}{Box, G.~E.}, \bibinfo{author}{Jenkins, G.~M.}, \bibinfo{author}{Reinsel, G.~C.} \& \bibinfo{author}{Ljung, G.~M.}
\newblock \emph{\bibinfo{title}{Time series analysis: forecasting and control}} (\bibinfo{publisher}{John Wiley \& Sons}, \bibinfo{year}{2015}).

\bibitem{hyndman2018forecasting}
\bibinfo{author}{Hyndman, R.~J.} \& \bibinfo{author}{Athanasopoulos, G.}
\newblock \emph{\bibinfo{title}{Forecasting: principles and practice}} (\bibinfo{publisher}{OTexts}, \bibinfo{year}{2018}).

\end{thebibliography}

% \clearpage

%------------------------------------
\section*{Supplementary Information}

\subsection*{Open source dataset}

To assess the real-world applicability of our quantum time series (QTS) model, we utilize a geophysical dataset comprising geopotential height measurements at the 500 hPa pressure level, commonly referred to as Z500, for the year 1992. This dataset is drawn from the WeatherBench benchmark archive and is provided in NetCDF format with a global spatial resolution of 5.625° × 5.625°. Geopotential height at 500 hPa is a central diagnostic in atmospheric science, capturing the large-scale mid-tropospheric flow patterns. It serves as a proxy for identifying synoptic-scale structures such as troughs, ridges, and jet streams and is routinely used in operational weather forecasting and atmospheric dynamics research.

We extract a zonal-mean time series at the equator (latitude = 0°), averaged longitudinally to produce a univariate time series. This reduction retains the core temporal variability of the equatorial circulation while allowing the model to focus on temporal learning in a physically interpretable context. This scalar series provides a challenging yet informative testbed for benchmarking quantum and classical forecasting approaches. The use of this dataset is motivated by its relevance in both numerical weather prediction and recent machine learning-based forecasting studies. It offers a scientifically meaningful benchmark for testing our QTS framework under climatologically significant conditions. The dataset is publicly available through the WeatherBench archive at \url{https://dataserv.ub.tum.de/s/m1524895}, and is documented in the work by Rasp et al.~\cite{rasp2020weatherbench}.

\subsection*{Discussion on Error}

The performance of the quantum time series model is affected by both classical and quantum sources of noise. These sources can be broadly categorized into input-level noise in the data and hardware-induced imperfections arising from quantum circuit execution.

\subsubsection*{Classical noise in input data}

The synthetic dataset used in our simulations is derived from an AR(1) process with additive Gaussian white noise:
\begin{equation}
    x_t = \phi x_{t-1} + \varepsilon_t, \qquad \varepsilon_t \sim \mathcal{N}(0, \sigma^2).
\end{equation}
This noise simulates typical real-world uncertainties, including sensor inaccuracies and thermal fluctuations, and introduces baseline stochasticity that affects both classical and quantum prediction tasks.

\subsubsection*{Quantum hardware-induced errors}

In a practical implementation, quantum circuits are affected by several layers of hardware-level noise.
Quantum hardware introduces several sources of error that can influence the accuracy of circuit-based time series prediction. Gate errors arise from imperfections in the implementation of quantum operations, particularly in multi-qubit gates such as CNOTs, which introduce both coherent and incoherent noise during state preparation and evolution. Decoherence further limits circuit fidelity, as qubits lose quantum coherence over time through relaxation ($T_1$) and dephasing ($T_2$) processes, potentially degrading the circuit's ability to retain and manipulate temporal patterns before measurement. Additionally, measurement noise affects the reliability of classical readout, where bit-flip errors in the final state assignment can distort the inferred output distribution. These combined effects necessitate careful circuit design and error-aware execution strategies, especially when deploying quantum time series models on near-term hardware.

In our simulations, we use an ideal (noise-free) backend to evaluate the theoretical capability of the proposed model. However, real hardware implementation would need to incorporate these noise models.

\subsection*{Error mitigation strategies}

To suppress the influence of quantum noise on model performance, a range of mitigation strategies can be employed. Designing shallow circuits using hardware-efficient gate decompositions helps minimize exposure to decoherence and gate infidelity. Optimal qubit placement and transpilation tailored to the specific backend can further reduce error rates by avoiding the most error-prone hardware paths. Statistical uncertainty from sampling can be addressed by executing circuits over a large number of shots, thereby stabilizing the output distribution. In addition, post-processing methods such as readout error calibration, zero-noise extrapolation, and measurement unfolding can be applied to adjust for systematic noise in measurement outcomes. When integrated with the quantum time series framework, these techniques collectively enhance the robustness and predictive reliability of the model under realistic hardware constraints.

\subsection*{Complexity Analysis}\label{complexity_dis_details}

In this section, we rigorously analyze the computational complexity of the proposed Quantum Time Series (QTS) models and benchmark them against prominent classical time series forecasting techniques, including ARIMA, Temporal Convolutional Networks (TCNs), and Transformer-based models. The comparison is structured around three primary complexity dimensions: the number of data points used for training (data complexity), the number of tunable parameters (parameter complexity), and the total computational overhead, particularly in the context of quantum circuit depth and gate count (gate or model complexity). The QTS models utilize a quantum circuit encoding strategy where each input time series segment is mapped to a quantum state via a set of single-qubit rotations $R_y(\theta_i)$ applied to $n$ qubits. These qubits collectively encode a window of $n$ classical values, meaning the model learns patterns over $\mathcal{O}(n)$ input points. Since $n = \log_2 N$, where $N$ denotes the size of the classical dataset, this suggests a logarithmic dependence on the training data volume.

\begin{lemma}\label{traininig_lemma}
Let $n$ be the number of qubits and $N = 2^n$ the size of the classical dataset. Then, the QTS model requires only $\mathcal{O}(\log N)$ data points for training.
\end{lemma}

\begin{proof}[Note on Lemma. \ref{traininig_lemma}]
Each quantum circuit iteration encodes $n$ classical values into $n$ qubits using a sliding window. Since this window has fixed size $n$ and slides one step per forecast point, the training process for $F$ forecast steps only requires $n + F$ values. As $n = \log_2 N$, the number of training points scales logarithmically with the dataset size $N$. Thus, the QTS model's training data requirement is $\mathcal{O}(\log N)$.
\end{proof}

\begin{lemma}\label{param_com}
The parameter complexity and gate complexity of the QTS model are both $\mathcal{O}(n)$.
\end{lemma}

\begin{proof}[Note on Lemma~\ref{param_com}]
The parameter vector $\vec{\theta}$ includes one parameter per qubit, resulting in $n$ rotation parameters for an $n$-qubit system and thus a parameter complexity of $\mathcal{O}(n)$. Regarding the gate complexity, the base QTS circuit uses $n$ single-qubit $R_y(\theta)$ rotation gates. When forward entanglement is included, it adds $n-1$ CNOT gates connecting adjacent qubits, and with the addition of cross entanglement, up to $\lfloor n/2 \rfloor$ more CNOT gates are introduced, linking qubits at distance two. These entangling gates introduce temporal correlations across the circuit without increasing the number of trainable parameters, which remain fixed at $\mathcal{O}(n)$. Consequently, the total gate count—including both single- and two-qubit gates—scales linearly with the number of qubits, giving an overall gate complexity of $\mathcal{O}(n)$. Since each circuit is executed $M_1$ times (shots) to generate the necessary measurement statistics for prediction, the resulting time complexity per forecast step is $\mathcal{O}(M_1 \cdot n)$.
\end{proof}

Table~\ref{tab:complexity} summarizes the asymptotic complexity comparisons of the models studied and detailed below.

\begin{table}[h!]
\centering
\caption{Complexity comparison across time series forecasting models.}
\label{tab:complexity}
\begin{tabular}{|l|c|c|c|}
\hline
\textbf{Model} & \textbf{Training Complexity} & \textbf{Parameter Complexity} & \textbf{Model/Gate Complexity} \\
% \hline
% AR(p), ARIMA & $\mathcal{O}(N)$ & $\mathcal{O}(p)$ & $\mathcal{O}(N \cdot p^2)$ \\
\hline
AR($p$) & $\mathcal{O}(N \cdot p^2)$ & $\mathcal{O}(p)$ & $\mathcal{O}(N \cdot p^2)$ \\
ARIMA($p,d,q$) & $\mathcal{O}(N \cdot (p+q)^2)$ & $\mathcal{O}(p+q)$ & $\mathcal{O}(N \cdot (p+q)^2)$ \\
\hline
\hline
TCN Ref.\cite{bai2018empirical,vaswani2017attention} & $\mathcal{O}(E \cdot L \cdot k \cdot N \cdot d^2)$ & $\mathcal{O}(L \cdot d^2)$ & $\mathcal{O}(L \cdot k \cdot N \cdot d^2)$ \\
Transformer Ref.\cite{vaswani2017attention} & $\mathcal{O}(E \cdot L \cdot N^2 \cdot d)$ & $\mathcal{O}(L \cdot d^2)$ & $\mathcal{O}(L \cdot N^2 \cdot d)$ \\
\hline
\hline
QTS ($R(\theta)$) & $\mathcal{O}(\log N)$ & $\mathcal{O}(\log N)$ & $\mathcal{O}(M_1\log N)$ \\
QTS ($R(\theta)$ + Fwd) & $\mathcal{O}(\log N)$ & $\mathcal{O}(\log N)$ & $\mathcal{O}(M_1\log N)$ \\
QTS ($R(\theta)$ + Fwd + Cross) & $\mathcal{O}(\log N)$ & $\mathcal{O}(\log N)$ & $\mathcal{O}(M_1\log N)$ \\
\hline
\hline
\end{tabular}
\end{table}

\subsubsection*{Complexity Analysis of AR and ARIMA Models}

Autoregressive (AR) and autoregressive Integrated Moving Average (ARIMA) models are widely used classical methods for univariate time series forecasting. The AR model of order $p$ predicts the current value of a time series as a linear combination of its past $p$ values, while the ARIMA model incorporates differencing (to handle non-stationarity) and a moving average component for error correction Ref.\cite{box2015time}. Given a training sequence of length $N$, the AR($p$) model uses least squares estimation or Yule-Walker equations to estimate $p$ coefficients. This typically involves solving a system of linear equations, leading to a computational cost of $\mathcal{O}(N \cdot p^2)$ when implemented via ordinary least squares. In terms of parameter complexity, the model learns $p$ coefficients, yielding a space complexity of $\mathcal{O}(p)$.

ARIMA models extend this by incorporating differencing and potentially a moving average of order $q$. Estimating ARIMA($p,d,q$) involves optimization over $(p+q)$ parameters and differencing the sequence $d$ times, often requiring iterative maximum likelihood estimation (MLE). This increases the training complexity to roughly $\mathcal{O}(N \cdot (p+q)^2)$ Ref.\cite{hyndman2018forecasting}, depending on convergence criteria and the specific estimation algorithm. Unlike deep learning models, AR and ARIMA models do not require iterative training over epochs or a large number of parameters, but their reliance on the full sequence length $N$ for each prediction step and optimization phase leads to scalability limitations in large datasets. Furthermore, they assume linearity in time dependencies, limiting their modelling capacity for nonlinear or highly complex patterns. For AR models, the computational complexity depends on the order \( p \), representing the number of lag terms used for forecasting. The training complexity for AR is \( \mathcal{O}(N \cdot p^2) \), primarily due to solving the least-squares problem for estimating coefficients over \( N \) data points. The parameter complexity is \( \mathcal{O}(p) \), as there are \( p \) coefficients to learn. For ARIMA models, which extend AR by incorporating differencing (of order \( d \)) and moving average components (of order \( q \)), the training complexity becomes \( \mathcal{O}(N \cdot (p+q)^2) \). Nevertheless, the parameter complexity remains \( \mathcal{O}(p+q) \), and the model complexity scales similarly with the data length and the squared total number of parameters. These linear models serve as strong baselines due to their simplicity and interpretability for various time-series forecasting, especially without much non-linearity in the data pattern.

\subsubsection*{Complexity of Temporal Convolutional Networks (TCNs)}

Temporal Convolutional Networks (TCNs) apply causal and dilated convolutions to capture long-range dependencies in time series data. Based on the analysis in~Ref.\cite{vaswani2017attention,bai2018empirical}, a convolutional layer with kernel size $k$ and hidden dimension $d$ operating over $N$ time steps exhibits a per-layer computational complexity of $\mathcal{O}(k \cdot N \cdot d^2)$. When employing $L$ such layers and training the network over $E$ epochs, the overall training complexity scales as $\mathcal{O}(E \cdot L \cdot k \cdot N \cdot d^2)$. The model's parameter complexity is dominated by the convolutional filters and is given by $\mathcal{O}(L \cdot d^2)$, assuming the same hidden dimension across layers.

Additionally, due to the use of dilated convolutions, TCNs offer a maximum path length between input and output of $\mathcal{O}(\log_k N)$, which is favourable compared to recurrent models. However, this benefit comes at the cost of a higher per-layer arithmetic workload, particularly when $d$ and $k$ are large. As such, while TCNs offer improved parallelism and temporal receptive fields, their complexity in both data and model size remains linear in $N$.

\subsubsection*{Complexity of Transformer Models}

Transformer architectures~Ref.\cite{vaswani2017attention} model temporal dependencies via self-attention mechanisms, enabling every position in the input sequence of length $N$ to directly attend to every other. According to the complexity analysis provided in Table~1 of~Ref.\cite{vaswani2017attention}, the computational cost per layer for full self-attention is $\mathcal{O}(N^2 \cdot d)$, where $d$ denotes the representation dimension of the model. Since Transformers support parallel processing across positions, they require only $\mathcal{O}(1)$ sequential operations and achieve a maximum path length of $\mathcal{O}(1)$, enabling efficient learning of long-range dependencies. For a model with $L$ attention layers trained over $E$ epochs, the overall training complexity scales as $\mathcal{O}(E \cdot L \cdot N^2 \cdot d)$. The parameter complexity is approximately $\mathcal{O}(L \cdot d^2)$, accounting for the attention weight matrices and feedforward transformations. Despite their strong expressivity and scalability in parallel hardware, Transformers become computationally expensive for large $N$ due to their quadratic time and memory requirements in sequence length.

\subsubsection*{Complexity of Entanglement-based Quantum Time Series Model}

Table~\ref{tab:complexity} highlights the efficiency of Quantum Time Series (QTS) models, which achieve logarithmic scaling in both data and parameter complexity relative to the total number of time steps \( N \). This stems from encoding \( N \) classical time steps into \( n = \log N \) qubits via amplitude encoding. As a result, training and parameter complexity remain at \( \mathcal{O}(\log N) \), considerably lower than classical models that typically require \( \mathcal{O}(N) \) or more. Model complexity, including gate depth and measurement repetitions \( M_1 \), scales as \( \mathcal{O}(M_1 \log N) \), accounting for single-qubit rotations and fixed-depth entangling layers using forward and cross \( \mathrm{CX} \) gates. These gates enhance expressivity without increasing parameter count. In contrast, classical models such as ARIMA, TCNs, and Transformers incur greater computational costs, often linear or quadratic in \( N \). Thus, QTS offers a promising approach for forecasting in low-resource settings through compact encoding and shallow, fixed-parameter circuits. While we do not claim performance superiority over classical methods, our results show that temporal patterns can be learned using minimal data, entanglement, and parameterized gates. Further optimization—including improved learning strategies, error mitigation, and deeper circuit architectures—may enhance quantum advantage in future implementations.

\subsection*{Entanglement-informed variance propagation}

We now formalize the role of entanglement in modulating predictive uncertainty through quantum correlations:

\begin{lemma}\label{lemma2}
Let $\{x_{t-i}\}_{i=1}^n$ be a set of past observations, each sampled as $x_{t-i} \sim \mathcal{N}(\mu, \sigma^2)$ and normalized to $[0,1]$, then encoded into qubit states via $R_y(\theta_{t-i})$ with $\theta_{t-i} = 2\sin^{-1}(x_{t-i})$. Let $|\Psi_t\rangle = \bigotimes_i R_y(\theta_{t-i})|0\rangle$ be the initial product state and $|\Psi_{\text{ent}}\rangle = U_{\text{ent}}|\Psi_t\rangle$ be the final state after applying entanglement. Then the variance of any observable $\mathcal{O}$ measured on $|\Psi_{\text{ent}}\rangle$ depends on the structure of $U_{\text{ent}}$ and introduces statistical correlations up to the entanglement range.
\end{lemma}

\begin{proof}[Note on Lemma. \ref{lemma2}]
In the absence of entanglement, each qubit contributes independently to the measurement distribution, resulting in factorized variances across outcomes. The inclusion of entangling gates, such as nearest-neighbour or cross CNOTs, entangles qubit amplitudes and induces correlated interference terms in the final state. These correlations propagate input variance nonlinearly through the circuit, modulating the final measurement probabilities $P(k)$ and hence altering the predictive statistics of $\hat{x}_{t+1}$.
\end{proof}

\end{document}